\documentclass[12pt,a4paper]{article}
\usepackage[utf8]{inputenc}
\usepackage[T1]{fontenc}
\usepackage[ngerman,english]{babel}
\usepackage{graphicx}
\usepackage{subfig}
\usepackage{epstopdf}
\usepackage{float}
\usepackage{latexsym}
\usepackage{amsmath,amssymb,amsthm}
\usepackage{amsfonts}
\usepackage[normalem]{ulem}
\usepackage{pdfpages}
\usepackage[font=footnotesize]{caption}
\usepackage{setspace}
\usepackage{paralist}
\usepackage{enumitem}
\usepackage{fancyhdr}
\usepackage{geometry}
\usepackage[numbers]{natbib}
\usepackage{float}
\usepackage{algorithm}
\usepackage{algpseudocode}
\usepackage{soul}

\newtheorem{Theorem}{Theorem}[section]
\newtheorem{Lemma}[Theorem]{Lemma}
\newtheorem{Corollary}[Theorem]{Corollary}

\theoremstyle{definition}
\newtheorem{Definition}[Theorem]{Definition}

\newtheorem{Remark}[Theorem]{Remark}
\newtheorem{Example}{Example}

\newcommand{\C}{\mathbb{C}}

\newcommand{\N}{\mathbb{N}}
\newcommand{\K}{\mathbb{K}}
\newcommand{\ord}{\textrm{ord}}

\newcommand{\cc}{\mathcal{C}(F)}

\newcommand{\yy}{\tilde{y}}
\newcommand{\cP}{\mathcal{P}}
\algnewcommand\algorithmicinput{\textbf{Input:}}
\algnewcommand\INPUT{\item[\algorithmicinput]}
\algnewcommand\algorithmicoutput{\textbf{Output:}}
\algnewcommand\OUTPUT{\item[\algorithmicoutput]}
\newcommand\para{\vspace*{2mm}}
\newcommand\mult{\mathrm{mult}}

\allowdisplaybreaks

\makeatletter
\newenvironment{breakablealgorithm}
{
	\begin{center}
		\refstepcounter{algorithm}
		\hrule height.8pt depth0pt \kern2pt
		\renewcommand{\caption}[2][\relax]{
			{\raggedright\textbf{\ALG@name~\thealgorithm} ##2\par}%
			\ifx\relax##1\relax 
			\addcontentsline{loa}{algorithm}{\protect\numberline{\thealgorithm}##2}%
			\else 
			\addcontentsline{loa}{algorithm}{\protect\numberline{\thealgorithm}##1}%
			\fi
			\kern2pt\hrule\kern2pt
		}
	}{
		\kern2pt\hrule\relax
	\end{center}
}
\makeatother

\begin{document}\selectlanguage{english}

\pagestyle{fancy}
\fancyhf{}
\renewcommand{\headrulewidth}{0pt}

\title{Solving First Order Autonomous Algebraic Ordinary Differential Equations by Places}
\author{Sebastian Falkensteiner\thanks{Research Institute for Symbolic Computation (RISC), Johannes Kepler University Linz, Austria.
		Email: falkensteiner@risc.jku.at} \and J.Rafael Sendra\thanks{Research Group ASYNACS. Dpto. de Física y Matemáticas, Universidad de Alcalá, Madrid, Spain.
		Email: Rafael.sendra@uah.es}}
\date{\empty}
\maketitle

\fancyfoot[C]{\thepage}

\begin{abstract}
	Given a first-order autonomous algebraic ordinary differential equation, we present a method for computing formal power series solutions by means of places. We provide an algorithm for computing a full characterization of possible initial values, classified in terms of the number of distinct formal power series solutions extending them. In addition, if a particular initial value is given, we present a second algorithm that computes all the formal power series solutions, up to a suitable degree, corresponding to it. Furthermore, when the ground field is the field of the complex numbers, we prove that the computed formal power series solutions are all convergent in suitable neighborhoods.
\end{abstract}

\noindent \textbf{keyword}
Algebraic autonomous differential equation,
algebraic curve,
local pa\-ra\-me\-tri\-za\-tion, place,
formal power series solution, analytic solution.

\section{Introduction} \label{sec:introduction}
A first-order algebraic ordinary differential equation (AODE) is a polynomial relation among a function and its derivatives. If the independent variable does not explicitly appear in the relation, the AODE is called autonomous. In this paper we are looking for formal power series solutions of first-order autonomous AODEs. More precisely, we are looking for truncations of formal power series up to a suitable degree, so that the truncation can be extended uniquely to a formal solution.

The problem of finding formal power series solutions of AODEs has a long history and has been extensively studied in the literature. Denef and Lipshitz consider AODEs of any order and their general formal power series solutions (see \cite{Denef1984}). As a consequence of their results, one can make an ansatz of unknown coefficients, plug it into the differential equation and compare coefficients. We will refer to this process as the direct method. A comparison of this approach with our method is illustrated in Example \ref{ex:direct}.

A method to compute generalized formal power series solutions, i.e. power series with real exponents, can be found in \cite{GrigorievSinger1991}. There, the authors introduce a parametric version of the Newton polygon method and use it to study generalized formal power series solutions of AODEs. For a more detailed insight into this method we refer to \cite{Cano2005} and \cite{CanoFortuny2009}.

In \cite{Aroca2000} the recursive Newton-Puiseux method for first-order (not necessarily autonomous) AODEs is presented. The author gives necessary conditions in each step for the existence of the next coefficient of the Puiseux series solutions; see Proposition 2.1. in \cite{Aroca2000}.
If the degree in $y'$ is one, the method is indeed algorithmic and the existence of Puiseux solutions is guaranteed; see Theorem 3.8. Moreover, the obtained solutions, for the first-order and first-degree AODE, are shown to be  convergent in the sense of Puiseux series; see Theorem 3.10. In addition, these results can be extended to the dicritical case (see \cite{Cano2005}).


There are contributions to obtain other types of solutions such as rational general solutions; see e.g. \cite{feng2004rational}, \cite{feng2006polynomial}. Extensions to algebraic general solutions and first-order non-autonomous AODEs can be found in \cite{aroca2005algebraic}, \cite{RISC4106}, \cite{RISC4381} and \cite{RISC5589}.

In this paper we compute all formal power series solutions of a first-order autonomous AODE, in the sense explained at the beginning of this section. In particular, we use the known bounds for computing places to ensure termination. Up to our knowledge, this cannot be achieved by any of the other methods.
Our approach strongly depends on geometric methods. More precisely, we associate a plane curve to the given first-order autonomous AODE, and we use the places of this curve to analyze the existence, and actual computation, of formal power series solutions.

The structure of the paper is as follows. Section \ref{sec-preliminaries} is devoted to recall the preliminary theory on algebraic curves to be used throughout the paper.
In Section \ref{sec-main} we show that every non-constant formal power series solution defines a place of the associated curve. Moreover, we give a sufficient and necessary condition on a place of the curve to contain in its equivalence class a formal power series solution of the original differential equation.
From this characterization, since all elements in the class are related by means of the substitution of formal power series of order one, the formal power series solution can be determined.
Using this strategy, and choosing as ground field the field of complex numbers, we prove that every formal power series solution is convergent in a certain neighborhood of zero.

Using the previous ideas, and some intermediate results which are presented in Section \ref{sec:critical}, in Section \ref{sec-alg} we outline two algorithms that are illustrated by examples. The first algorithm provides a characterization of all possible initial tuples, i.e. the first two coefficients of the formal power series solutions, of the equation. More precisely, the algorithm decomposes the points of the affine plane into sets such that the number of formal power solutions, centered at each point in the set, is the same. The output of the algorithm is described in finite terms.
On the other hand, for any particular initial tuple, the second algorithm determines the coefficients of all solutions, starting with this initial tuple, up to a suitable order such that it is distinguished from the other solutions.

\section{Preliminaries} \label{sec-preliminaries}
In this section we recall some basic notions on differential algebra and algebraic geometry that will be used throughout this paper. For further details, we refer to \cite{kolchin1973differential}, \cite{ritt1950differential} and \cite{walker1950algebraic}.

Let $\K$ be an algebraically closed field of characteristic zero and $n \in \N$. We denote by $\K\{y\}$ the ring of differential polynomials in $y$ with coefficients in $\K$. A differential polynomial is of order $n \geq 0$ if the $n$-th derivative $y^{(n)}$ is the highest derivative appearing in it.

In the following, we consider irreducible first-order autonomous algebraic ordinary differential equations (AODEs) of the form
\begin{equation}\label{EQ-AODE}
F(y,y')=0,
\end{equation}
where $F \in \K[y,y'] \setminus \K[y]$. For technical reasons, that will be clear below and in subsequence sections, we will also assume that $F$ is not of the form $y'-\lambda$ with $\lambda\in \K$. We observe that none of the assumptions implies any loss of generality.
Indeed, if $F$ is reducible, one might consider each irreducible component of the factorization and proceed accordingly. If $F$ does not depend on the first derivative, it is just an algebraic equation, and finally, the case $y'=\lambda$ is trivial.


Our goal is to decide the existence and, in the positive case, actually determine the formal power series solutions of (\ref{EQ-AODE}), in the sense explained in Section \ref{sec:introduction}. Our strategy to approach this problem consists in analyzing an algebraic curve associated to the given AODE. This algebraic curve is introduced as follows.
Let $F$ be as in (\ref{EQ-AODE}). We consider $y$ and $y'$ as independent variables, let us say $y$ and $z$. So, $F(y,z)\in \K[y,z]$. $F$ is non-constant because $F \notin \K[y]$, and hence it defines a plane affine algebraic curve, say $\mathcal{C}(F)$, over $\K$; namely,
$$\mathcal{C}(F)=\{(a,b) \in \K^2 ~|~ F(a,b)=0\}.$$
We call $\mathcal{C}(F)$ the \textit{corresponding curve} of $F$. Observe that because of our assumptions $\cc$ is irreducible and it is neither a vertical line nor the $y$-axis.


Let $\K[[t]]$ be the ring of formal power series in the indeterminate $t$ and coefficients in $\K$. Assume that
\begin{equation*} \label{solution}
\tilde{y}=\sum_{i \geq 0}c_it^i \in \K[[t]]
\end{equation*}
is a formal power series solution of (\ref{EQ-AODE}).
Then, substituting $0$ for $t$, we get that $(c_0,c_1)\in \cc.$ So, we get a first necessary condition on the initial tuples of the solution. Furthermore, the constant solutions of (\ref{EQ-AODE}) are the first coordinates of the intersection points of $\cc$ and the line $z=0$.
In order to deal with the other coefficients of the non-constant formal power solutions, one may proceed as follows. Let $(c_0,c_1)\in \cc$ and let the coefficients $c_i$, with $i>1$, in $\yy$, be undetermined.
Then, for every $k \geq 1$, there exists a differential polynomial $R_k \in \K\{y\}$ of order at most $k$ such that
\begin{equation*}
F^{(k)}=S_F \cdot y^{(k+1)}+R_k,
\end{equation*}
where $S_F=\frac{\partial F}{\partial y^{'}}$ is the separant of $F$ (see e.g.  \cite{ritt1950differential}[page 6]).
Then, if $S_F(c_0,c_1) \neq 0$, we can use the above formula to recursively determine $c_k$, for $k>1$, namely
\begin{equation}\label{eq-separant}
c_{k}=-\dfrac{R_{k-1}(c_0,\ldots,(k-1)!c_{k-1})}{k!S_F(c_0,c_1)}.
\end{equation}
However, if the separant does vanish at $(c_0,c_1)$, a formal power series solution cannot be derived using the previous reasoning.
We observe that, under our assumptions, the set of points $(c_0,c_1)$, for which the above direct method does not work, is finite and non-empty. Indeed, since $F \notin \K[y]$, then $S_F$ is not identically zero.
Furthermore, $\deg_{y,z}(S_F(y,z))=\deg_{y,z}(F)-1$ and since $F$ is irreducible, by B\'ezout's theorem (see e.g. \cite{walker1950algebraic}), the number of common zeros of $\{F(y,z)=0,S_F(y,z)=0\}$, counted with multiplicity, is equal to $\deg_{y,z}(F)(\deg_{y,z}(F)-1)$, which is not $0$ since $F$ is not of the form $y'-\lambda$.


Therefore, in order to complete the analysis still finitely many points on $\cc$ need to be analyzed. For this purpose, we will use the places of $\mathcal{C}(F)$. Let us recall this notion. A pair $\mathcal{P}=(A,B) \in K[[t]]^2$ is called a \textit{local parametrization} of $\mathcal{C}(F)$ if $F(A,B)=0$ holds, and not both power series $A$ and $B$ are constant.
$\textbf{c}=(A(0),B(0))$ is called the \textit{center} of the local parametrization $\mathcal{P}$.
The \textit{order} of $A$ at the center $\textbf{c}$ (similarly for $B$), denoted by $\ord_{\textbf{c}}(A)$, is defined as $\ord(A-A(0))$. $\mathcal{P}$ is called \textit{reducible}, if there exists another local parametrization $\mathcal{P}^* \in \K[[t]]^2$ and $r>1$ such that $\mathcal{P}=\mathcal{P}^*(t^r)$.
Otherwise $\mathcal{P}$ is called \textit{irreducible}. We also recall that the units in $\K[[t]]$ are exactly the power series of order $0$.


Now, let $S \in \K[[t]]$ with $1 \leq \ord(S)<\infty$ and let $\mathcal{P} $ be a local parametrization. By using the usual substitution for formal power series, $\mathcal{P}(S)$ is again a local parametrization of the same curve with the same center. Finally, we recall that two local parametrizations $\mathcal{P}_1,\mathcal{P}_2 \in \K[[t]]^2$ of the same curve $\mathcal{C}$ are called \textit{equivalent} if there exists $S \in \K[[t]]$ with $\ord(S)=1$, such that $\mathcal{P}_1(S)=\mathcal{P}_2.$ This is an equivalence relation. A \textit{place}  is the equivalence class of an irreducible local parametrization of  the curve. The center of the place is  the common center point of all local parametrizations in the equivalence class.  In the sequel, in case of non ambiguity, we will not distinguish between places and irreducible local parametrizations.


We finish this preliminary section recalling some properties and concepts related to places that will be used later.
Let $\mathcal{P}=(a_0+a_nt^n+\cdots,b_0+b_mt^m+\cdots)$ be the representative of a place of $\cc$ with $a_nb_m \neq 0$ and $n,m>0$. Then, the tangent vector $\bar{v}$ of $\mathcal{C}(F)$ at $(a_0,b_0)$ through the place defined by $\mathcal{P}$ is (see Section 5.3. in \cite{walker1950algebraic})
\begin{equation}\label{eq-tg}
\bar{v}=
\begin{cases}
(a_n,b_n), \text{ if } n=m,\\
(a_n,0), \text{ if } n<m,\\
(0,b_m), \text{ if } n>m.
\end{cases}
\end{equation}
We recall that the order of the place $(A,B)$ centered at $\mathbf{c}$ is $\min\{\ord_{\mathbf{c}}(A),\ord_{\mathbf{c}}(B) \}$.
The order of a place is related to the multiplicity of its center as follows. The multiplicity of $\mathcal{C}(F)$ at $\mathbf{c}$ is equal to the sum of the orders of all places of $\cc$ centered at $\mathbf{c}$; see e.g. Theorem 5.8(i),(ii) in \cite{walker1950algebraic}.


A simple point of $\cc$ with a tangent parallel to one of the axes is called a \textit{ramification point}. We distinguish between ramification points w.r.t. $z$, where the tangent is parallel to the $z$-axis, and ramification points w.r.t. $y$, where the tangent is parallel to the $y$-axis.
Finally, let $\mathbf{c}\in \cc$ be simple, and let $(A,B)$ define a place at $\mathbf{c}$; see above the relationship between the multiplicity of a point and the orders of the places centered at the point. Then by (\ref{eq-tg}) we get that $\mathbf{c}$ is a $z$-ramification point if and only if $\ord_{\textbf{c}}(A)>\ord_{\textbf{c}}(B)=1$, and it is a $y$-ramification point if and only if $\ord_{\textbf{c}}(B)>\ord_{\textbf{c}}(A)=1$.

\section{Solution Places} \label{sec-main}
In this section we state our main result that shows how the places of $\cc$ and the non-constant formal power solutions of (\ref{EQ-AODE}) relate. We start showing that non-constant formal power solutions always define places of the associated curve $\cc$.

\begin{Lemma} \label{LEM:para}
	Let $\yy=\sum_{i \geq 0}c_it^i$ be a non-constant solution of (\ref{EQ-AODE}). Then
	\begin{enumerate}
		\item[(i)] $(\tilde{y},\tilde{y}')$ defines a place of $\cc$ with center at $(c_0,c_1)$; and
		\item[(ii)] A place of $\cc$ can be defined by at most one formal power series solution of (\ref{EQ-AODE}).
	\end{enumerate}
	We will refer to $(\yy(t),\yy'(t))$ as the place defined by the solution $\yy(t)$.
\end{Lemma}
\begin{proof}
	(i) $F(\tilde{y},\tilde{y}')=0$ and $\tilde{y}$ is non-constant by assumption. So $\mathcal{P}=(\tilde{y},\tilde{y}')$ defines a local parametrization at the center $(\tilde{y}(0),\tilde{y}'(0))=(c_0,c_1).$
	For proving the irreducibility of $\mathcal{P}$, let $$\tilde{y}=c_0+c_st^s+ \cdots ~\text{ with } c_s \neq 0,s>0$$ and let us assume that there exists another local parametrization $\mathcal{P}^*=(A,B)$ and an $r>1$ such that $\mathcal{P}=\mathcal{P}^*(t^r).$
	We express $A,B$ as $$A=a_0+a_nt^n+ \cdots, \quad B=b_0+b_mt^m+ \dots,$$ where $a_n,b_m \neq 0$.
	Then $$
	\begin{array}{l} \tilde{y}=c_0+c_st^s+ \cdots=A(t^r)=a_0+a_nt^{rn}+ \cdots,  \\
	\tilde{y}'=c_1+sc_st^{s-1}+ \cdots=B(t^r)=b_0+b_mt^{rm}+ \cdots.
	\end{array}$$ Hence, $s=rn$ and if $s=1,$ then $r=1$ follows. If $s>1$, we get $s-1=rm$ and therefore, $r(n-m)=1$ implying $r=1$. Thus, $\mathcal{P}$ is irreducible in both cases.
	
	\noindent (ii) Let $\yy_1,\yy_2\in \K[[t]]$ be solutions of (\ref{EQ-AODE}) defining the same place of $\cc$. Then, there exists $S\in \K[[t]]$ of order 1, such that $(\yy_1(S),\yy_1'(S))=(\yy_2,\yy_2')$.
	Therefore,
	\[ \yy_2'=(\yy_1(S))'=\yy_1'(S)S'=\yy_1'(S). \]
	Since $(\yy_1,\yy_1')$ is a local parametrization, we have that $\yy_1'\neq 0$, and hence $S'=1$. So, using that $\ord(S)=1$, we get that $S=t$. Thus $\yy_2=\yy_1$.
\end{proof}

Based on the previous lemma, we introduce the following concept.

\begin{Definition}
	We say that a place of $\cc$ is a \textit{solution place} if it is definable by a non-constant formal power series solution of (\ref{EQ-AODE}).
\end{Definition}

Taking into account that the orders of the  coordinates of all the local parametrizations in a place are the same, a necessary condition for a place $(A,B)$ for being a solution place is that $\ord(A')=\ord(B)$. This motivates the next concept.

\begin{Definition}
	Let $(A,B) \in \K[[t]]^2$ be a place of $\cc$. We say that $(A,B)$ is \textit{order-suitable} if $\ord(A')=\ord(B)$. Moreover, we call $\textbf{c} \in \mathcal{C}(F)$ a \textit{suitable center}, if there exists an order-suitable place centered at $\textbf{c}$.
\end{Definition}


The next theorem shows that the two previous notions are equivalent.

\begin{Theorem} \label{THM:main}
	Let $\cP$ be a place of $\cc$. Then $\cP$ is a solution place if and only if $\cP$ is a order-suitable place.
\end{Theorem}
\begin{proof}
	Clearly a solution place is an order-suitable place.
	
	Conversely, let $\cP$ be order-suitable and let $(A,B) \in \K[[t]]^2$ be a representative of $\cP$.
	Let $S$ be a formal power series of order at least 1, with unknown coefficients. Since $\cP$ is of suitable order, $A'(S)S'=B(S)$ can be expressed as $$(a_kS^k+a_{k+1}S^{k+1}+\cdots) \cdot S'= b_kS^k+b_{k+1}S^{k+1} \cdots$$ with $a_kb_k \neq 0.$
	Now factoring out $S^k$ on both sides and multiplying by the inverse of the series $a_k+a_{k+1}S+\cdots$ we get the equivalent differential equation
	\begin{equation} \label{EQU:repara}
	S'=(a_k+a_{k+1}S+\cdots)^{-1}(b_k+b_{k+1}S+\cdots).
	\end{equation}
	Following the idea of the method of Limits, see e.g. Section 12.2. in \cite{Ince1926}, and comparing coefficients we obtain a unique formal power series solution $S=\sum_{i \geq 1}s_it^i$ with $\ord(S)=1$. For the constant coefficients in (\ref{EQU:repara}) we have $s_1=\frac{b_k}{a_k} \neq 0.$ The $i$-th coefficient with $i \geq 1$ on the left hand side is equal to $(i+1)s_{i+1}$. On the right hand side, due to the well-known formulas for the substitution, product and inversion of formal power series, we obtain a polynomial in $s_1,\ldots,s_{i+1}$ (and $a_k,\ldots,a_{k+i},b_k,\ldots,b_{k+i}$). Hence, the coefficients of $S$ are recursively given.
	Thus, $F(A(S),(A(S))')=F(A(S),B(S))=0$ and $A(S)$ is a solution of (\ref{EQ-AODE}). Let us assume that $A(S)$ is constant. This is the case if and only if $A$ is constant. Then also $B(S)=(A(S))'=0$ and therefore $B=0$ in contradiction to the assumption that $(A,B)$ is a local parametrization. Thus $\cP$ is a solution place.
\end{proof}

\begin{Remark}\label{Rem:1to1}
	In the previous proof, we have seen that if one starts from a local parametrization $(A,B)$ with coefficients in a subfield $\mathbb{L}$ of $\K$, then $S$ has coefficients in $\mathbb{L}$. Hence, the non-constant formal power series solution $A(S)$ belongs to $\mathbb{L}[[t]]$.
\end{Remark}

\begin{Theorem}\label{Thm:analytic}
	Let $\K=\C$. Then, all formal power solutions of (\ref{EQ-AODE}) are analytic.
\end{Theorem}
\begin{proof}
	The statement is trivial for the case of constant solutions. Let $\yy$ be a non-constant formal power solution of (\ref{EQ-AODE}).
	By Lemma \ref{LEM:para}, there is a unique solution place $\cP$ defined by $\yy$. On the other hand, let $(A(t),B(t))$ be the irreducible local parametrization of $\cP$ generated by the Newton-Puiseux method (see e.g. \cite{walker1950algebraic}).
	By \cite{casas-alvero_2000}[Theorem 1.7.2], we get that $A,B$ are analytic.
	Furthermore, from Section 12.2 in \cite{Ince1926}, we know that (\ref{EQU:repara}) has a convergent solution $S$ in a certain neighborhood of its center if $A$ and $B$ are convergent.  Thus, $\yy=A(S)$ is analytic.
\end{proof}

\section{Critical Sets}\label{sec:critical}
In Section \ref{sec-main} we have seen that the non-constant formal power series solutions of (\ref{EQ-AODE}) are related to the places in $\cc$ that are of suitable order.
To proceed algorithmically with the ideas in Section \ref{sec-main}, we need to ensure that almost all, i.e., all but finitely many, points in $\cc$, are suitable centers; note that, as already observed, points in $\K^2\setminus \cc$ cannot generate solutions. For this purpose, we introduce the notion of critical set.

\begin{Definition}\label{def:critical}
	We say that $\mathcal{S}\subset \cc$ is a \textit{critical set} if $\mathcal{S}$ is finite, and every point $\cc\setminus \mathcal{S}$ is the center of at least one solution place.
\end{Definition}

We observe that, because of Theorem \ref{THM:main}, a point in $\cc$ is the center of at least one  solution place if and only if it is of suitable order for (\ref{EQ-AODE}). In the following, we prove that there always exist critical sets for the equation (\ref{EQ-AODE}). For this, we first prove the next technical lemma.

\begin{Lemma}\label{lem:charac-order-suitable}
	Let $\mathbf{c}=(c_0,c_1) \in \cc$ and let $(A,B)$ be a local parametrization of $\cc$ centered at $\mathbf{c}$. $\mathbf{c}$ is an order-suitable center if and only if both (i) and (ii) hold, where
	\begin{enumerate}
		\item[(i)] if $c_1 \neq 0$ then $\ord_{\mathbf{c}}(A)=1$.
		\item[(ii)] if $c_1=0$ then $\ord_{\mathbf{c}}(A)=\ord_{\mathbf{c}}(B)+1$.
	\end{enumerate}
\end{Lemma}
\begin{proof}
	Let $A=c_0+a_r t^r+\cdots, B=c_1+a_st^s+\cdots$ where $a_rb_s\neq 0$. If $\mathbf{c}$ is order-suitable, then $\ord(A')=\ord(B)$. Therefore, if $c_1=0$ then $r-1=s$, and if $c_1\neq 0$ then $r=0$. Thus, (i) and (ii) hold.
	Conversely, if (i) and (ii) hold then
	$\ord(B)=0=r-1$ if $c_1\neq 0$ and $\ord(B)=c=r-1$ if $c_1=0$.
	Thus, $\ord(A')=\ord(B)$. So $\mathbf{c}$ is order-suitable.
\end{proof}

\begin{Corollary} \label{COR:suff}
	Let $\mathbf{c}=(c_0,c_1)\in \cc$ with $c_1\neq 0$. If $\mathbf{c}$ is simple and not of $z$-ramification, then it is a suitable center.
\end{Corollary}
\begin{proof}
	Let $(A,B)$ be a local parametrization of $\mathcal{C}(F)$ centered at $\textbf{c}$. Since $\textbf{c}$ is simple, $\min(\ord_\textbf{c}(A),\ord_\textbf{c}(B))=1$ (see Section \ref{sec-preliminaries}). Since $\textbf{c}$ is not of $z$-ramification, $\ord_\textbf{c}(A)=1$ (see Section \ref{sec-preliminaries}). By assumption, $c_1 \neq 0$ and therefore, by Lemma \ref{lem:charac-order-suitable}, $\textbf{c}$ is a suitable center.
\end{proof}

We conclude the section by proving the existence of critical sets. For this, we will use the following notation. If $\mathcal{F}$ is a nonempty subset of $\K[y,z]$ we denote by $\mathbb{V}(\mathcal{F})$ the affine variety of $\K^2$ defined by $\mathcal{F}$; i.e. the zero-set of $\mathcal{F}$ over $\K$.

\begin{Theorem}\label{th-critical-set}
	$\mathbb{V}(\{F(y,z),z\})\cup\mathbb{V}(\{F(y,z),S_F(y,z)\}),$ where $S_F$ is the separant of $F$, is a critical set of $\cc$.
\end{Theorem}
\begin{proof}
	Let $\mathcal{S}=\mathbb{V}(\{F(y,z),z\})\cup\mathbb{V}(\{F(y,z),S_F(y,z)\})$. If $\mathbf{c}=(c_0,c_1)\in \cc\setminus \mathcal{S}$, then $F(\mathbf{c})=0$, $c_1\neq 0$, and $S_F(\mathbf{c})\neq 0$.
	Therefore, by Corollary \ref{COR:suff}, we get that $\mathbf{c}$ is a suitable center.
	
	The fact that $\mathcal{S}$ is finite follows by applying B\'ezout's theorem as follows: $F$ is irreducible, and by our general assumptions $F\neq z$. So, $\mathbb{V}(\{F(y,z),z\})$ is finite.
	Moreover, $\deg_{y,z}(S_F)<\deg_{y,z}(F)$ and $S_F$ is not zero, because under our hypotheses $F$ does depend on $z$. Thus, $\mathbb{V}(\{F(y,z),S_F(y,z)\})$ is finite too.
\end{proof}

\begin{Remark}\label{rem:puntos-eje}
	We observe that the places centered at the points in $\mathbb{V}(\{F,z\})\setminus \mathbb{V}(\{F,S_F\})$ are not solution places. Indeed, let $\mathbf{c}$ be in the previous set. Then, $c_1=0$, and $S_F(\mathbf{c})\neq 0$. So, by equation (\ref{eq-separant}), there exists a unique formal power solution of (\ref{EQ-AODE}) initialized at $\mathbf{c}$. But, $\yy=c_0$ is indeed a constant solution. Thus, the place centered at $\mathbf{c}$ is not a solution place.
\end{Remark}

\section{Algorithms and Examples} \label{sec-alg}

In this section, using the results from Section \ref{sec-main}, we derive two different algorithms  and illustrate them by some examples. In addition, we compare our method with the direct method (see Section \ref{sec:introduction}) by commenting on a main difference and illustrate it with an example.


Given $\mathbf{c}\in \K^2$, the first algorithm computes truncated expressions for all formal power series solutions of (\ref{EQ-AODE}), with $\mathbf{c} \in \K^2$ as initial tuple, so that the output truncation of any two different solutions do not coincide.
The second algorithm classifies all initial tuples $\mathbf{c}\in \cc$ by the number of distinct formal power series solutions initialized at $\mathbf{c}$.

In both algorithms, we focus on non-constant solutions. Note that the constant solutions are precisely the points on $\cc$ with a zero in the second component.

For this purpose, we will use two auxiliary algorithms. The first auxiliary algorithm, denoted by \textsf{LocalParametrization}$(F,\mathbf{c})$, provides a set containing the truncation of a  representative of each place centered at $\mathbf{c}\in \cc$, so that any two output truncated local parametrizations do not coincide. More precisely, the algorithm works at follows. After a suitable change of coordinates we may assume w.l.o.g. that $\mathbf{c}$ is the origin. In this situation, we apply to the irreducible polynomial $F\in K[y][z]$ the algorithm described in \cite{walker1950algebraic}[Section 3.2.] to compute the set $\mathcal{R}$ of all roots of $F$ up to a given order $N$, in the field of formal Puiseux series. Alternatively one may use the rational Puiseux expansion algorithm described in \cite{Duval1989}.
In \cite{Duval1989} (see also \cite{RISC4119}), upper bounds for the number of terms of the singular part of the Puiseux expansions in $\mathcal{R}$ are given; for instance $2(\deg_y(F)-1)\deg_z(F)+1$ is an upper bound.
So, $N$ can be taken bigger or equal to this quantity. In this situation, we recall that the positive order elements in $\mathcal{R}$ correspond to the places of $\cc$ centered at $\mathbf{c}$ (see Theorem 4.1. in \cite{walker1950algebraic}[page 107]).
Thus, the positive order elements in $\mathcal{R}$ provide (truncated) local parametrizations of the form $(c_0+\lambda t^k,B)$ for some $k \in \N, \lambda\in \K$, and $B \in \K[[t]]$.
Finally, analyzing equivalences among these (truncated) local parametrizations one gets the (truncated) places at $\mathbf{c}$ (see e.g. \cite{Duval1989} or \cite{RISC4119}).

For the second auxiliary algorithm, we recall that, for a given $N \in \N$, in the proof of Theorem \ref{THM:main} we have seen how to compute the coefficients $s_1,\dots,s_N$ of the reparametrization $S=\sum_{i \geq 1}s_it^i$, namely with the ansatz of unknown coefficients and coefficient comparison.
Let $(A,B)$ be a solution place with $A=c_0+a_rt^r+\cdots, B=c_1+b_st^s+\cdots$ and $a_rb_s\neq 0$. We note that only $\{a_r,\ldots,a_{r+i}\}$ and $\{c_1,b_s\ldots,b_{s+i-1}\}$ are required to compute the $i$-th coefficient of $S$. Hence, in order to compute (truncated) solutions up to order $N$, the first $N$ coefficients of the local parametrization have to be computed. We will refer to this constructible method as \textsf{Reparametrization}$(A,B,N)$.

If $\K=\C$, starting from the representation of a local parametrization as mentioned above, one can decide whether this branch is real or not and transform it into a real one by the substitution $t \mapsto \lambda t$, where $\lambda$ is a root of unity; see \cite{alonso1992}[Section 5.10]. Moreover, by Remark \ref{Rem:1to1}, we can check whether the solutions are real for a given initial tuple.


Now, we proceed with the description of the proposed algorithms.

\begin{breakablealgorithm}
	\caption{} \label{ALG:GivenInitialCoeffs}
	\begin{algorithmic}[1]
		\INPUT A first-order autonomous AODE $F(y,y')=0$, satisfying the general assumptions of (\ref{EQ-AODE}), and an initial tuple $\textbf{c}=(c_0,c_1) \in \K^2$.
		\OUTPUT A set consisting of the truncations of all non-constant formal power series solutions with initial tuple $(c_0,c_1)$ such that for every two different solutions the output truncations are different.
		\State Set $\mathcal{U}_{\textbf{c}}=\emptyset$.
		\If {$F(\textbf{c})=0$} compute $\mathcal{R}:=\textsf{LocalParametrization}(F,\mathbf{c})$.
		\For {every $(A,B)\in \mathcal{R}$}
		\If {$\ord(A')=\ord(B)$}
		\State $S:=\textsf{Reparametrization}(A,B,\max\{\deg(A'),\deg(B)\})$
		\State Set $\mathcal{U}_{\textbf{c}}=\mathcal{U}_{\textbf{c}} \cup \{A(S)\}$.
		\EndIf
		\EndFor
		\EndIf \\
		\Return $\mathcal{U}_{\mathbf{c}}$.
	\end{algorithmic}
\end{breakablealgorithm}

\begin{Theorem}\label{thm:alg-1}
	Algorithm \ref{ALG:GivenInitialCoeffs} is correct.
\end{Theorem}
\begin{proof}
	Since $\mathcal{R}$ is a finite set and \textsf{LocalParametrization} and \textsf{Reparametrization} terminate, the algorithm terminates.
	Let $\mathcal{S}$ be the set of all non-constant formal power solutions of (\ref{EQ-AODE}) at $\mathbf{c}$.
	By Theorem \ref{THM:main}, $\mathcal{U}_{\mathbf{c}}$ contains the truncations of the elements in $\mathcal{S}$.
	
	Now, let us prove that $\#(\mathcal{S})=\#(\mathcal{U}_{\mathbf{c}})$.
	We first observe that $\mathcal{R}$ contains as many elements as places centered at $\mathbf{c}$, and each of them is a truncation of a place.
	
	The computed truncations can be expressed in the form $(A,B)=(c_0+t^r,c_1+b_kt^k+\dots+b_{k+N}t^{k+N})$ and are distinct due to \cite{RISC4119}[Lemma 1 and Lemma 2], where the multiplicity of $\cc$ at $\mathbf{c}$ is an upper bound for the calculated coefficients until the first distinct appears. Clearly, $N \geq \mult_{\mathbf{c}}(\cc)$ and therefore all computed (truncations of) local parametrizations are different.
	
	Let $(\tilde{A},\tilde{B})=(c_0+t^{\tilde{r}},c_1+\tilde{b}_{\tilde{k}}t^{\tilde{k}}+\dots+\tilde{b}_{\tilde{k}+N}t^{\tilde{k}+N})$ be another computed truncation. If $\tilde{r} \neq r$, then obviously the truncated solutions obtained via them are different. So let us assume that $\tilde{r}=r$ and set $\tilde{k}=k$.
	Moreover, let $m$ be the first index such that $b_{s+m} \neq \tilde{b}_{s+m}$. Then the $(m+1)$-st coefficient $s_{m+1}$ and $\tilde{s}_{m+1}$ of the reparametrizations $S$ and $\tilde{S}$ are the first different coefficients; note that in (\ref{EQU:repara}) by the recursive construction for $s_{m+1}$ and $\tilde{s}_{m+1}$ all terms on the right hand side are the same except for the only summand where $b_{k+m}$ ($\tilde{b}_{k+m}$ respectively) appears, namely $\frac{b_{k+m}s_1^m}{a_k}$ and $\frac{\tilde{b}_{k+m}s_1^m}{a_k}$.
	Therefore, the coefficients
	$$[t^{r+m}]A(S)=rs_1^{r-m}s_{m+1} + (\text{ terms in } s_1,\ldots,s_m)$$   $$[t^{r+m}]\tilde{A}(\tilde{S})=r\tilde{s}_1^{r-m}\tilde{s}_{m+1} + (\text{ terms in } \tilde{s}_1,\ldots,\tilde{s}_m)$$ are different.
\end{proof}

\begin{breakablealgorithm}
	\caption{} \label{ALG:NoInitialCoeffs}
	\begin{algorithmic}[1]
		\INPUT A first-order autonomous AODE $F(y,y')=0$ as in (\ref{EQ-AODE}).
		\OUTPUT A partition of the points of $\mathcal{C}(F)$ such that the number of non-constant formal power series solutions with this point as initial tuple is the same.
		\State Compute the variety $\mathcal{B}=\mathbb{V}(F,z) \cup \mathbb{V}(F,S_F)$ containing all critical points.
		\State For every critical point apply Algorithm \ref{ALG:GivenInitialCoeffs} to decide how many non-constant solutions are centered at it.
		\State Collect in the sets $\mathcal{A}_i$ the critical points providing $i$ non-constant solutions for $i \geq 0$.
		\State Replace $\mathcal{A}_1$ by $\mathcal{A}_1 \cup (\mathcal{C}(F) \setminus \mathcal{B})$.\\
		\Return $\mathcal{A}_0,\mathcal{A}_1,\ldots$ which are non-empty.
	\end{algorithmic}
\end{breakablealgorithm}

We want to emphasize that in Algorithm \ref{ALG:GivenInitialCoeffs} the set $\mathcal{A}_1$ has always infinitely many elements. However, we can give back this set in closed form as we indicate in the last step of the algorithm.

\begin{Theorem}\label{thm:alg-2}
	Algorithm \ref{ALG:NoInitialCoeffs} is correct.
\end{Theorem}
\begin{proof}
	By Theorem \ref{th-critical-set}, $\mathcal{B}$ is either empty or finite. So the algorithm terminates.
	Now, the correctness of Algorithm \ref{ALG:NoInitialCoeffs} follows from Def. \ref{def:critical} and Theorem \ref{thm:alg-1}.
\end{proof}


\begin{Example}\label{ex-1}
	Consider the AODE $$F(y,y')=y'^2-y^3-y^2=0$$ with $F \in \C[y,y'].$
	The corresponding curve $\mathcal{C}(F)$ is a rational cubic with a double point at the origin (see Fig. \ref{fig-C(F)} left).
	We apply Algorithm \ref{ALG:NoInitialCoeffs} to decompose the points of $\cc$ depending on the existence of non-constant analytic solutions.
	In Step 1, we get $\mathcal{B}=\mathbb{V}(\{z^2-y^3-y^2,z\})=\{(0,0),(-1,0)\}$. Let $\mathbf{c}_1:=(0,0)$ and $\mathbf{c}_2:=(-1,0)$. In Step 2 we apply Algorithm \ref{ALG:GivenInitialCoeffs} to the critical points.
	
	\para
	
	\noindent For $\mathbf{c}_1$ we get two places locally parametrized by $(t,t+\mathcal{O}(t^2))$ and $(t,-t+\mathcal{O}(t^2))$. So, $\mathcal{R}=\{(t,t),(t,-t)\}$. Since none of the places are order-suitable, we get $\mathcal{U}_{\mathbf{c}_1}=\emptyset.$
	
	\para
	
	\noindent
	For $\mathbf{c}_2$ we get the place $(A,B):=(t^2-1,t-t^3)$ which is order-suitable. Therefore, we apply Algorithm \textsf{Reparametrization} to $(A,B,3)$. We obtain $S=\frac{t}{2}-\frac{t^3}{24}+\frac{t^5}{240}.$ Thus, $\mathcal{U}_{\mathbf{c}_2}=\{A(S)\}=\{ -1+\frac{t^2}{4}-\frac{t^4}{24}\}$.
	
	\para
	
	Therefore, the output of the algorithm is $\mathcal{A}_0:=\{\mathbf{c}_1\}$, $\mathcal{A}_1:=\cc\setminus \{\mathbf{c}_1\}$, which means that:
	(1) at every point in $\mathcal{A}_1$ there exists exactly one non-constant formal power solution that is, indeed, analytic; (2) there is no non-constant analytic solution at $\mathbf{c}_1$; (3) $\{0,-1\}$ are the only constant solutions.
	
	
	Now, to further illustrate the method, we take $\mathbf{c}=(1,\sqrt{2}) \in \mathcal{A}_1$. A local parametrization at $\mathbf{c}$ is $(A,B)=(1+t,\sqrt{2}+\frac{5\sqrt{2}t}{4}+\frac{7\sqrt{2}t^2}{32}-\frac{3\sqrt{2}t^3}{128}+\mathcal{O}(t^4)).$
	Applying Algorithm \textsf{Reparametrization} to $(A,B,3)$ we get  $S=\sqrt{2}t+\frac{5t^2}{4}+\frac{2\sqrt{2}t^3}{3}+\mathcal{O}(t^4)$ that provides the solution
	$$A(S)= 1+\sqrt{2}t+\frac{5t^2}{4}+\frac{2\sqrt{2}t^3}{3}+\mathcal{O}(t^4).$$
\end{Example}

\begin{Example}\label{ex-2}
	Consider the following AODE
	$$F(y,y')=((y'-1)^2+y^2)^3-4(y'-1)^2y^2=0$$
	with $F \in \C[y,y']$.
	The corresponding curve $\mathcal{C}(F)$ is a rational degree $6$ curve with a non-ordinary  singularity at $\mathbf{c}=(0,1)$ (see Fig. \ref{fig-C(F)} right). We apply Algorithm \ref{ALG:NoInitialCoeffs}. In Step 1, we get
	$ \mathcal{B}= \{(0,1),(\alpha,0), (\frac{4\beta}{9},\gamma)\}$
	where $\alpha^6+3\alpha^4-\alpha^2+1=0$, $\beta^2=3,$ and $27\gamma^2-54\gamma +19=0$.
	\begin{figure}[ht]
		\centerline{ \includegraphics[width=4cm]{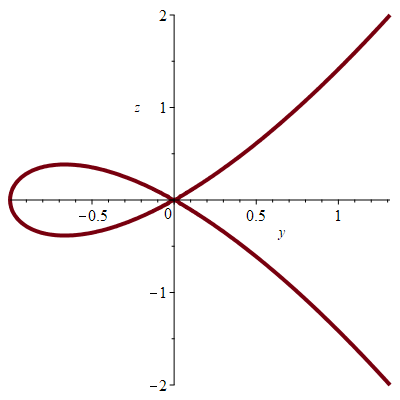} \hspace*{4mm}
			\includegraphics[width=4cm]{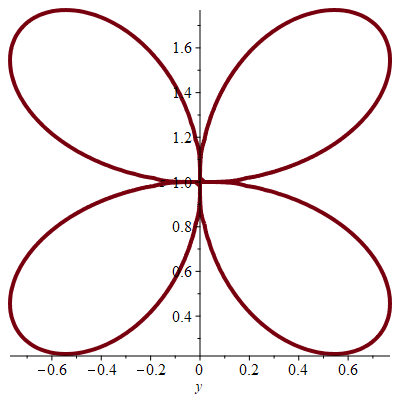} }
		\caption{Plot of $\mathcal{C}(F)$ of Example \ref{ex-1} (left) and Example \ref{ex-2} (right) } \label{fig-C(F)}
	\end{figure}
	In Step 2 we apply Algorithm \ref{ALG:GivenInitialCoeffs} to the critical points.
	
	\noindent Let $\mathbf{c}_1:=(0,1)$. We get the following four places
	\[ \begin{array}{ll}
	\cP_1 = (t^2,1+\sqrt{2}t+\mathcal{O}(t^2)) &
	\cP_2 = (-t^2,1-\sqrt{2}t+\mathcal{O}(t^2)) \\ \\
	\cP_3 = (t,1+\frac{t^2}{2}+\frac{3t^4}{16}+\mathcal{O}(t^6)) &
	\cP_4 = (t,1-\frac{t^2}{2}-\frac{3t^4}{16}+\mathcal{O}(t^6)).
	\end{array}
	\]
	So, $\mathcal{R}=\{
	(t^2,1+\sqrt{2}t),(-t^2,1-\sqrt{2}t),(t,1+\frac{t^2}{2}+\frac{3t^4}{16}), (t,1-\frac{t^2}{2}-\frac{3t^4}{16})\}.$ Observe that only
	$\cP_3,\cP_4$ are order-suitable. We apply to them Algorithm \textsf{Reparametrization}. One obtains
	$S=t+\frac{1}{6}t^3+\frac{17}{240}t^5$ for $\cP_3$, and $ S=t-\frac{1}{6}t^3+\frac{17}{240}t^4$ for $\cP_4$,
	that generate two solutions at $(0,-1)$, namely (truncated)
	\[ \yy_1=t+\frac{t^3}{6}+\frac{17t^5}{240},\,\, \yy_2=t-\frac{t^3}{6}+\frac{17t^5}{240}.
	\]
	Thus, $\mathcal{U}_{\mathbf{c}_1}=\{\yy_1,\yy_2\}$.
	
	\noindent
	Let $\mathbf{c}_\alpha:=(\alpha,0)$ where $\alpha^6+3\alpha^4-\alpha^2+1=0$. We get the place
	\[ \left(\alpha+t,\left(\frac{11}{19}\alpha^5+\frac{36}{19}\alpha^3+\frac{4}{19}\alpha\right)t+ \left(\frac{953}{13718}\alpha^4+\frac{4119}{13718}\alpha^2+\frac{1847}{6859}\right)t^2+ \cdots \right) \]
	which is not order-suitable (compare to Remark \ref{rem:puntos-eje}).
	Thus, $\mathcal{U}_{\mathbf{c}_\alpha}=\emptyset$.
	
	\noindent
	Let $\mathbf{c}_{\beta,\gamma}:=\left(\frac{4\beta}{9},\gamma\right)$, where $\beta^2=3,$ and $ 27\gamma^2-54\gamma +19=0$. We get the
	place
	$ ( \frac{4\beta}{9}-t^2,\gamma-\frac{\sqrt[4]{27}}{3}t+\cdots)$
	which is not order-suitable.
	Thus, $\mathcal{U}_{\mathbf{c}_{\beta,\gamma}}=\emptyset$.
	
	\para
	
	Therefore, the output of Algorithm \ref{ALG:NoInitialCoeffs} is $\mathcal{A}_0:=\{\mathbf{c}_\alpha,\mathbf{c}_{\beta,\gamma}\}$, $\mathcal{A}_1:=\cc \setminus \mathcal{B}$, and $\mathcal{A}_2:=\{\mathbf{c}_1\}$, where $\alpha,\beta,\gamma$ are as above. Moreover, the constant solutions are $\{\alpha\}$. To summarize: (1)
	At $\mathbf{c}_1$ there are two non-constant analytic solutions. (2) At each point in $\cc \setminus \mathcal{B}$ there is exactly 1 non-constant analytic solution.
	(3) At $\mathbf{c}_\alpha$ there is no non-constant analytic solution.
	(4) At $\mathbf{c}_{\beta,\gamma}$ there is no non-constant analytic solution.
\end{Example}

In the last part of the section we compare our method with the direct method (see Section \ref{sec:introduction}). The direct method consists in computing the solutions of $F(c_0,c_1)= \cdots = F^{(k)}(c_0,\dots,c_{k+1})=0$ in $\K$. However, finding an upper bound $M \in \N$ such that every solution of $F(c_0,c_1)= \cdots = F^{(M)}(c_0,\dots,c_{M+1})=0$ can be extended to a solution $\tilde{y}=\sum_{i \geq 0}c_it^i$ of $F(y,y')=0$ is in general an unsolved problem. This difficulty does no appear in our method as the following example illustrates:

\begin{Example}\label{ex:direct}
	Let us consider the family of irreducible AODEs $$G_m(y,y')=(y'-1)^2-y^{2m+1}=0,$$ where $m \in \N$.
	Set $F_1=(y'-1)^2$ and $F_2=y^{2m+1}$. Then $G_m=F_1-F_2$.
	It is straightforward to see that for each $k \geq 1$, the differential polynomial $F_1^{(k)}$ is the sum of a linear $\C^{*}$-combination of terms of the form $y^{(i)}y^{(j)}$, with $i+j=k+2$ and $i,j>1$, and the term $2(y'-1)y^{(k+1)}$.\\
	Moreover, for each $1 \leq k \leq 2m$, the differential polynomial $F_2^{(k)}$ is a linear $\C$-combination of terms of the form $y^{2m-i}f_i(y',\dots,y^{(k-i)})$, where $f_i$ is a non-constant monomial and $0 \leq i \leq k-1$. Hence, for $0 \leq k \leq m$ and $\textbf{c}=(0,1)$ the equations $$G_m^{(2k)}(\mathbf{c},c_2,\dots,c_{k+1})=0$$ hold if and only if $c_{k+1}=0$.
	However, since $$G_m^{(2m+1)}(\mathbf{c},0,\dots,0)=-(2m+1)! \neq 0,$$ $\mathbf{c}$ cannot be extended to a formal power series solution.
	This is exactly the result obtained by Theorem \ref{THM:main}: for every $m$ the curve $\mathcal{C}(G_m)$ has a tangent vector parallel to the $z$-axis at $\mathbf{c}$ and therefore is not a suitable center.
\end{Example}


\noindent \textbf{Acknowledgements.} The authors are supported by the Spanish Ministerio de Econom\'{\i}a y Competitividad, by the European Regional Development Fund (ERDF), under the MTM2017-88796-P. The first author is also supported by the strategic program "Innovatives O\"O 2020" by the Upper Austrian Government and by the Austrian Science Fund (FWF): P 31327-N32.

\appendix

\bibliographystyle{authordate1}

\end{document}